\documentclass[a4paper,11pt,reqno]{elsarticle}
\usepackage[latin1]{inputenc}
\usepackage{mathrsfs}
\usepackage{dsfont}
\usepackage{hyperref}
\usepackage{amsmath}
\usepackage{amssymb}
\usepackage{amsthm}
\usepackage{amsfonts}
\usepackage{amstext}
\usepackage{amsopn}
\usepackage{amsxtra}
\usepackage{mathrsfs}
\usepackage{dsfont}
\usepackage{esint}
\usepackage{enumitem}
\usepackage{dsfont}
 \usepackage{graphicx}

\newtheorem{thm}{Theorem}
\newtheorem{lemma}[thm]{Lemma}

\newtheorem{remark}[thm]{Remark}

\newcommand{\R}{\mathbb{R}}

\newcommand{\C}{\mathbb{C}}

\renewcommand\phi{\varphi}

\renewcommand{\geq}{\geqslant}
\renewcommand{\leq}{\leqslant}

\renewcommand{\tilde}{\widetilde}

\newcommand{\be}{\begin{equation}}
\newcommand{\ee}{\end{equation}}
\newcommand{\bq}{\begin{equation}}
\newcommand{\eq}{\end{equation}}

\newcommand{\eps}{\varepsilon}
\usepackage{color}

\newtheorem{definition}{Definition}
\title{Stationary solutions for the 2D critical Dirac equation with Kerr nonlinearity}

\author[W. Borrelli]{William Borrelli}
\address{Universit\'e Paris-Dauphine, PSL Research University, CNRS, UMR 7534, CEREMADE, F-75016 Paris, France} 
\ead{borrelli@ceremade.dauphine.fr}
\journal{Journal of Differential Equations}

\date{\today}

\DeclareRobustCommand{\SkipTocEntry}[5]{}

\begin{document}

\begin{abstract}
In this paper we prove the existence of an exponentially localized stationary solution for a two-dimensional cubic Dirac equation. It appears as an effective equation in the description of nonlinear waves for some Condensed Matter (Bose-Einstein condensates) and Nonlinear Optics  (optical fibers) systems. The nonlinearity is of Kerr-type, that is of the form $\vert \psi\vert^{2}\psi$ and thus not Lorenz-invariant. We solve compactness issues related to the critical Sobolev embedding $H^{\frac{1}{2}}(\R^{2},\C^{2})\hookrightarrow L^{4}(\R^{2},\C^{4})$ thanks to a particular radial ansatz. Our proof is then based on elementary dynamical systems arguments. 
\end{abstract}

\maketitle

\tableofcontents

\section*{Introduction}
The Dirac equation has been widely used to build relativistic models of particles(see the survey paper \cite{els}). 

Recently, it made its appearance in Condensed Matter Physics. New two-dimensional materials possessing Dirac fermions as low-energy excitations have been discovered, the most famous being the graphene \cite{graphene} (2010 Nobel Prize in Physics awarded to A.Geim and K. Novoselov). Those \textit{Dirac materials}, possess unique electronic properties which are consequence of the Dirac spectrum. They range from superfluid phases of $^{3}$He, high-temperature d-wave superconductors, graphene to topological insulators (see \cite{diracmaterials,diracfermions,introdiracmaterials} and references therein). Particular symmetries control the appearance of Dirac points. Time-reversal symmetry in topological insulators and sublattice symmetry in graphene \cite{introdiracmaterials} are some examples. In the paper \cite{FWhoneycomb} the authors rigorously proved existence and stability of Dirac cones for honeycomb Schr\"{o}dinger operators, under fairly general assumptions. 

The possibility of finding three-dimensional materials exhibiting a Dirac spectrum has also recently gained attention in the Physics community \cite{diracmaterials}.

In contrast to the case of many metals and doped semi-conductors, where nearly free quasi-particles obeying the Schr\"{o}dinger equation with an effective mass represent a very accurate approximation for low energy-excitations, for Dirac materials an accurate description is provided by the Dirac hamiltonian
$$ H=-iv_{F}(\vec{\sigma}\cdot\nabla)+mv_{F}^{2}\sigma_{3}$$
where the speed of light is replaced by the Fermi velocity $v_{F}$ and $m$ is an effective mass.

If $m=0$ the dispersion relation is linear (i.e. a cone), in contrast with the parabolic dispersion of metals or semiconductors. This case includes graphene monolayers \cite{introdiracmaterials}. 

The case of a non-vanishing mass term ($m\neq0$) corresponds to a gap at the Fermi level. It describes, for instance, a monolayer of boron-nitride or graphene bilayers (\cite{diracfermions}). It has been experimentally proved that placing boron-nitride in contact with graphene leads to the appearance of a non-zero mass, thus creating an energy gap.

Furthermore, using arguments from \cite{FTWprotected} and \cite{binding} a multiscale expansion shows that applying a suitable electric field (formally) opens a gap in the effective Dirac hamiltonian for the graphene, in the case of wavefunctions spectrally concentrated around a Dirac point. In the recent paper \cite{binding} the authors showed the existence of a gap for honeycomb Schr\"{o}dinger operators in the strong-binding regime, when an electric potential that breaks the $\mathcal{PT}$-symmetry (parity+time-inversion) is applied. 


An important model in nonlinear optics and in the description of macroscopic quantum phenomena (see \cite{boseeinstein},\cite{nonlinearoptics}) is the \textit{cubic Schr\"{o}dinger / Gross-Pitaevskii equation}:

\begin{equation}\label{eq:GP}
i\partial_{t}\Psi=\left(-\Delta+V\right)\Psi+g\vert\Psi\vert^{2}\Psi
\end{equation}
where $g$ is a parameter that measures the scattering length and the cubic term is a mean field interaction or a Kerr-nonlinear term due to a variable refractive index, according to the model.

The above equation appears, for instance, in the description of Bose-Einstein condensates.

If $V$ is a honeycomb potential, the low-energy effective operator around a Dirac point is the Dirac operator (see \cite{introdiracmaterials}) :

\begin{equation}
\left(-\Delta+V\right) \longmapsto\mathcal{D}:=-ic\left(\vec{\sigma}\cdot\nabla\right)
\end{equation}
Note that it acts on two-components spinors 
$$ \psi=\begin{pmatrix}\psi^{1} \\ \psi^{2} \end{pmatrix}\in\mathbb{C}^{2}$$
since the honeycomb lattice is a superposition of two triangular Bravais lattices. In this case the spinor encodes the isospin of the sublattices, rather than the proper spin of the electron (see \cite{introdiracmaterials}). 

As remarked above, applying a suitable electric potential or placing the material on a substrate results in an additional mass term. Thus the effective equation reads as

\begin{equation}\label{eq:timeDirac}
i\partial_{t}\Psi=\left(\mathcal{D}+m\sigma_{3}\right)\Psi+g\vert\Psi\vert^{2}\Psi
\end{equation} 
Our aim is to prove the existence of stationary solutions to (\ref{eq:timeDirac}) in the \textit{focusing} case, $g=-1$. Setting $$\Psi(x,t)=e^{-i\omega t}\psi(x)$$ with $0<\omega<m$, the equation rewrites as

\begin{equation}\label{equation}
\left(\mathcal{D}+m\sigma_{3}-\omega\right)\psi-\vert\psi\vert^{2}\psi=0
\end{equation}

The main result of this paper is the following
\begin{thm}\label{main}
Equation (\ref{equation}) admits a smooth localized solution, with exponential decay at infinity. 
\end{thm}
\begin{remark}
The result presented here is at odds with the case of the pseudo-relativistic operator $$\sqrt{-\Delta+m^{2}}\geq 0$$
Indeed, a simple Pohozaev-type argument shows that there is no smooth exponentially localized solution to the following equation
\begin{equation}
\left(\sqrt{-\Delta+m^{2}}\right)\psi-\omega\psi=\vert\psi\vert^{2}\psi\qquad\mbox{on}\quad\mathbb{R}^{2}
\end{equation}
with $0< \omega <m$. 

Thus the existence of solutions is related to the presence of the negative part of the spectrum of the Dirac operator (see next section).
\end{remark}

\begin{remark}
In this case the zero-energy corresponds to the Fermi level. Then there is no interpretation of the Dirac spectrum in terms of particles/antiparticles. Rather, the positive part of the spectrum corresponds to massive conduction electrons, while the negative one to valence electrons.
\end{remark}
\noindent\textbf{Acknowledgment.} 
The author wishes to thank \'{E}ric S\'{e}r\'{e} for his support.

\section{Preliminaries}
The Dirac operator is a first order differential operator formally defined in 2D (in the standard representation) as

\begin{equation}
\mathcal{D}_{m}=\mathcal{D}+mc^{2}\sigma_{3}:=-ic\hbar(\vec{\sigma}\cdot\nabla)+mc^{2}\sigma_{3}
\end{equation}
where $c$ denotes the speed of light, $m$ is the electron mass, $\hbar$ is the reduced Planck constant, $\vec{\sigma}\cdot\nabla:=\sigma_{1}\partial_{1}+\sigma_{2}\partial_{2}$ and the $\sigma_{k}$ are the Pauli matrices
\begin{equation}\label{eq:pauli} \sigma_{1}:=\begin{pmatrix} 0 \quad& 1 \\ 1 \quad& 0 \end{pmatrix}\quad,\quad \sigma_{2}:=\begin{pmatrix} 0 \quad& -i \\ i \quad& 0 \end{pmatrix} \quad,\quad \sigma_{3}:=\begin{pmatrix} 1 \quad& 0 \\ 0 \quad& -1 \end{pmatrix}\end{equation}
In this paper we shall work with a system of physical units such that $c=1$ and $\hbar=1$.

It is well known (see \cite{diracthaller}) that $\mathcal{D}_{m}$ is a self-adjoint operator on $L^{2}(\mathbb{R}^{2},\mathbb{C}^{2})$, with domain $H^{1}(\mathbb{R}^{2},\mathbb{C}^{2})$ and form-domain $H^{\frac{1}{2}}(\mathbb{R}^{2},\mathbb{C}^{2})$. 

Moreover, in Fourier domain $p=(p_{1},p_{2})$ the Dirac operator becomes the multiplication operator by the matrix 
$$ \widehat{\mathcal{D}}_{m}(p)=\begin{pmatrix}m \quad & p_{1}-ip_{2}\\ p_{1}+ip_{2}\quad &m \end{pmatrix}$$ 
so the spectrum is given by

\begin{equation}\label{eq:spectrum}
Spec(\mathcal{D}_{m})=(-\infty, -m]\cup[m, +\infty)
\end{equation}
where the gap is due to the mass term.

In this paper we focus on the following equation
\begin{equation}\label{eq:cubic}
\left(\mathcal{D}+m\sigma_{3}-\omega\right)\psi-\vert\psi\vert^{2}\psi=0 \qquad\mbox{on}\quad\mathbb{R}^{2},\quad\mbox{with}\quad 0<\omega<m
\end{equation}
whose weak solutions correspond to critical points of the following functional
\begin{equation}\label{eq:func}
\mathcal{L}(\psi):=\frac{1}{2}\int\langle\left(\mathcal{D}+m\sigma_{3}-\omega\right)\psi,\psi\rangle-\frac{1}{4}\int\vert\psi\vert^{4}
\end{equation}
defined for $\psi\in H^{\frac{1}{2}}(\mathbb{R}^{2},\mathbb{C}^{2})$.

The above functional is strongly indefinite, that is, it is unbounded both from above and below, even modulo finite dimensional subspaces. This is due to the unboundedness of $Spec(\mathcal{D})$. Several techniques have been introduced to deal with such situations (see for instance \cite{struwevariational}). 


Moreover, the main difficulty in our case is given by the lack of compactness of the Sobolev embedding $H^{\frac{1}{2}}(\mathbb{R}^{2},\mathbb{C}^{2})\hookrightarrow L^{4}(\mathbb{R}^{2},\mathbb{C}^{2})$. This implies the failure of some compactness properties used to prove linking results (see \cite{struwevariational} and references therein), due to the invariance by translations and scaling. 

In what follows we will only give a sketch of the compactness analysis for the above functional, referring to the mentioned papers for more details.

As we will see in the next section, equation (\ref{eq:cubic}) is compatible with a particular ansatz, leading us to work in the closed subspace
\be\label{subspace}
E=\left\{\psi\in H^{\frac{1}{2}}(\R^{2},\C^{2}):\psi(r,\vartheta)=\begin{pmatrix} v(r) \\ iu(r)e^{i\vartheta} \end{pmatrix}, u,v:(0,+\infty)\rightarrow \R \right\}
\ee
where $(r,\vartheta)$ are the polar coordinates of $x\in\R^{2}$.

Restricting the problem to the subspace $E$ breaks the invariance by translations, and thus to recover compactness one has to deal with the invariance by scaling only. The latter causes the so-called \textit{bubbling phenomenon}, that is, energy concentration associated to the appearance of blow-up profiles. In \cite{Isobecritical} Isobe analyzed the behavior of a generic Palais-Smale sequence for the critical Dirac equation on compact spin manifolds. The same can be done in our case.

Given a Palais-Smale sequence $\left(\psi_{n}\right)\subseteq H^{\frac{1}{2}}$ it easy to see that it is bounded, and thus we may suppose, up to extraction, that it weakly converges $$\psi_{n}\rightharpoonup\psi_{\infty}\in H^{\frac{1}{2}}.$$

Generally speaking, the invariance by scaling prevents the strong convergence and we have the profile decomposition
\be\label{profile}
 \psi_{n}=\psi_{\infty}+\sum^{N}_{k=1}\omega^{k}_{n}+o(1)\qquad\mbox{in}\quad H^{\frac{1}{2}}(\R^{2},\C^{2})
 \ee
where $N\in\mathbb{N}$ and $\omega^{k}_{n}$ is a properly rescaled $\mathring{H}^{\frac{1}{2}}(\R^{2},\C^{2})$-solution of the limit equation
$$ \mathcal{D}\varphi = \vert\varphi\vert^{2}\varphi$$

centered around points $a^{k}_{n}\rightarrow a^{k}\in\R^{2}$, as $n\rightarrow+\infty$, for $1\leq k\leq N$.

The \textit{bubbles} $\omega^{k}_{n}$ are in a finite number, since one can prove a uniform lower bound for their energy. Moreover, this implies that we have compactness only in a suitable energy range and gives a treshold value for the appearance of bubbles in min-max methods (see \cite{struwevariational}).

Then in terms of $L^{4}$-norms, there holds
 \be\label{concentration}
 \vert\psi_{n}\vert^{4} dx\rightharpoonup\vert\psi_{\infty}\vert^{4}dx+\sum^{N}_{k=1}\nu_{k}\delta_{a^{k}}
 \ee
weakly in the sense of measures. Here $\nu_{k}\geq0$ and the $\delta_{a^{k}}$ are delta measures concentrated at $a^{k}$.

Morever, since we are essentially working with radial functions, it's not hard to see that the blow-up can only occur at the origin, that is, we actually have
\be\label{concentrationorigin}
\vert\psi_{n}\vert^{4} dx\rightharpoonup\vert\psi_{\infty}\vert^{4}dx+\nu\delta_{0}
\ee
with $\nu\geq0$ and $\delta_{0}$ being the delta concentrated at the origin.

We thus conclude that in order to recover compactness for the variational problem one should be able to control the behavior of Palais-Smale sequences near the origin. 

However, our proof is based on a shooting method and thus not variational. In this case the concentration phenomenon (\ref{concentrationorigin}) manifests itself in the difficulty of controlling the behavior of solutions of the resulting dynamical system when initial data are large. This makes the analysis quite delicate and requires a careful asymptotic expansion of the solution, after a suitable rescaling (see section \ref{asymptotic}). 

We mention that the first rigorous existence result of stationary solutions for the Dirac equation via shooting methods is due to Cazenave and Vazquez \cite{cv}, who studied the Soler model for elementary fermions. Subsequently, those methods have been used to prove the existence of excited states \cite{exciteddirac} for the Soler model and in mean field theories for nucleons (see e.g. \cite{estebannodari},\cite{letreustnodari}, \cite{els} and references therein). We remark that a variational proof has been given by Esteban and S\'{e}r\'{e} in \cite{es}, under fairly general assumptions on the self-interaction. In particular, after a suitable radial ansatz, they prove a multiplicity result exploiting the Lorentz-invariance. Remarkably, their method works without any growth assumption on the nonlinearity.  However, the proof is designed to deal with the Lorentz-invariant form of the nonlinear term and is not applicable in our case. In \cite{dingwei} Ding and Wei proved an existence result for the 3D Dirac equation with a subcritical Kerr-type interaction. The case of a critical nonlinearity in 3D has been investigated by Ding and Ruf \cite{dingruf} in the semiclassical regime, using variational techniques. They take advantage of the presence of a negative potential to prove compactness properties. However, in this paper we deal with a \textit{critical} Kerr nonlinearity without additional assumptions and so we need to adopt a different strategy. 

\section{Existence by shooting method}\label{shooting}
To begin with, we first convert the equation into a dynamical system thanks to a particular ansatz. Then we will give some qualitative properties of the flow, particularly useful in understanding the long-time behavior of the system.

Passing to polar coordinates in $\R^{2}$ $(x,y)\mapsto ( r,\vartheta)$, the equation $$\left(\mathcal{D}+m\sigma_{3}-\omega\right)\psi-\vert\psi\vert^{2}\psi=0$$ reads as 

\begin{equation}
\left\{\begin{aligned}
 -e^{-i\vartheta}\left(i\partial_{r}+\frac{\partial_{\vartheta}}{r}\right)\psi_{2} &=\left(\vert\psi_{1}\vert^{2}+\vert\psi_{2}\vert^{2}\right)\psi_{1}-(m-\omega)\psi_{1} ,
\\
-e^{i\vartheta}\left(i\partial_{r}-\frac{\partial_{\vartheta}}{r}\right)\psi_{1} &=-\left(\vert\psi_{1}\vert^{2}+\vert\psi_{2}\vert^{2}\right)\psi_{2}-(m+\omega)\psi_{2}.
\end{aligned}\right.
\end{equation}
where $\psi=\begin{pmatrix} \psi_{1} \\ \psi_{2} \end{pmatrix}\in\mathbb{C}^{2} $, and this suggests the following ansatz (see \cite{cuevas}): 

\begin{equation}\label{ansatz}
\psi(r,\vartheta)=\begin{pmatrix} v(r)e^{iS\vartheta} \\ iu(r)e^{i(S+1)\vartheta} \end{pmatrix}
\end{equation}
with $u$ and $v$ real-valued and $S\in\mathbb{Z}$. In the sequel, we set $S=0$.

Plugging the above ansatz into the equation one gets 

\begin{equation}\label{radial}
\left\{\begin{aligned}
    \dot{u}+\frac{u}{r} &=(u^{2}+v^{2})v-(m-\omega)v \\ 
   \dot{v}&=-(u^{2}+v^{2})u-(m+\omega)u
\end{aligned}\right.
\end{equation}
Thus we are lead to study the flow of the above system. 

In particular, since we are looking for localized states, we are interested in solutions to (\ref{radial}) such that $$(u(r),v(r))\rightarrow(0,0)\qquad \mbox{as} \qquad r\rightarrow +\infty$$

In order to avoid singularities and to get non-trivial solutions, we choose as initial conditions $$u(0)=0\quad,\quad v(0)=\lambda\neq 0$$
Moreover, the symmetry of the system allows us to consider only the case $\lambda>0$. 

Studying the long-time behavior of the flow of (\ref{radial}) it is useful to introduce the following system
\begin{equation}
\label{hamiltonian}
\left\{\begin{aligned}
    \dot{u} &= (u^{2}+v^{2})v-(m-\omega)v \\ 
   \dot{v}&=-(u^{2}+v^{2})u-(m+\omega)u
\end{aligned}\right.
\end{equation}
Heuristically, (\ref{radial}) should reduce to (\ref{hamiltonian}) in the limit $r\rightarrow+\infty$ ($u$ being bounded), that is, dropping the singular term in the first equation. 

As one can easily check, (\ref{hamiltonian}) is the hamiltonian system associated with the function
\begin{equation}
\label{H}
H(u,v)=\frac{(u^{2}+v^{2})^{2}}{4}+\frac{m}{2}(u^{2}-v^{2})+\frac{\omega}{2}(u^{2}+v^{2})
\end{equation}
It's easy to see that the level sets of the hamiltonian $$\left\{H(u,v)=c\right\} $$ are compact, for all $c\in\mathbb{R}$, so that the flow is globally defined. 

The equilibria of the hamiltonian flow are the points
\be
(0,0), (0,\pm\sqrt{m-\omega})
\ee
and there holds
\be\label{equilibri}
H(0,0)=0,\quad H(0,\pm\sqrt{m-\omega})<0
\ee

Local existence and uniqueness of solutions of (\ref{radial}) are guaranteed by the following
\begin{lemma}
\label{existence}
Let $\lambda>0$. There exist $0<R_{\lambda}\leq+\infty$ and $(u,v)\in C^{1}([0,R_{\lambda}),\mathbb{R}^{2})$ unique maximal solution to (\ref{radial}), which depends continuously on $\lambda$ and uniformly on $[0,R]$ for any $0<R<R_{\lambda}$.
\end{lemma}
\begin{proof}
We can rewrite the system in integral form as
\begin{equation}\label{integral}
\left\{\begin{aligned}
   u(r) &= \frac{1}{r}\int^{r} _{0}sv(s)[u^{2}(s)+v^{2}(s)-(m-\omega)]ds \\ 
   v(r)&= \lambda-\int^{r}_{0}u(s)[(u^{2}(s)+v^{2}(s))+(m+\omega)]ds
\end{aligned}\right.
\end{equation}
where the r.h.s. is a Lipschitz continuous function. Then the claim follows by a contraction mapping argument, as in \cite{cv}.
\end{proof}

Given $\lambda>0$, define 
\begin{equation} 
H_{\lambda}(r):= H(u_{\lambda}(r),v_{\lambda}(r)) \quad,\quad r\in[0,R_{\lambda})
\end{equation}
where $(u_{\lambda},v_{\lambda})$ is the  solution of (\ref{radial}) such that $(u(0),v(0))=(0,\lambda)$. 

A simple computation gives
\begin{equation}
\label{decreasing}
 \dot{H}_{\lambda}(r)=-\frac{u^{2}_{\lambda}}{r}(m+\omega+u^{2}_{\lambda}(r)+v^{2}_{\lambda}(r))\leq0 \quad,\quad\forall r\in[0,R_{\lambda})
 \end{equation}
 
so that the energy $H$ is non-increasing along the solutions of (\ref{radial}). 

This implies that $\forall r\in[0,R_{x})$, $(u_{\lambda}(r),v_{\lambda}(r))\in \{H(u,v)\leq H(0,\lambda)\}$, the latter being a compact set. Thus there holds

\begin{lemma}
Every solution to (\ref{radial}) is global.
\end{lemma}
\begin{remark}
The above result is in contrast with the case of Lorentz-invariant models in 3D (\cite{exciteddirac}), where the energy has no definite sign and blow-up may occur.  
\end{remark}
The following lemma indeed shows that the solutions to (\ref{radial}) are close to the hamiltonian flow (\ref{hamiltonian}) as $r\rightarrow +\infty$. The proof is the same as the one given in \cite{cv}.
\begin{lemma}
\label{stability}
Let $(f,g)$ be the solution of (\ref{hamiltonian}) with initial data $(f_{0},g_{0})$. Let $(u^{0}_{n},v{0}_{n})$ and $\rho_{n}$ be such that 
$$\rho_{n}\xrightarrow{n\rightarrow+\infty}+\infty\qquad\mbox{and}\qquad(u_{n},v_{n})\xrightarrow{n\rightarrow+\infty} (f_{0},g_{0}) $$
Consider the solution of 
$$\left\{\begin{aligned}
    \dot{u}_{n}+\frac{u_{n}}{r+\rho_{n}} &= (u^{2}_{n}+v^{2}_{n})v_{n}-(m-\omega)v_{n}  \\ 
   \dot{v}_{n}&= -(u^{2}_{n}+v^{2}_{n})u_{n}-(m+\omega)u_{n}
\end{aligned}\right. $$
such that $u_{n}(0)=u^{0}_{n}$ and $v_{n}(0)=v^{0}_{n}$.

Then $(u_{n},v_{n})$ converges to $(f,g)$ uniformly on bounded intervals.
\end{lemma}
Since we know from (\ref{decreasing}) that the energy $H_{\lambda}$ decreases along the flow of (\ref{radial}) and that each solution is bounded, Lemma (\ref{stability}) allows us to conclude (see the proof of Lemma (\ref{decay})) that any solution must tend to an equilibrium of the hamiltonian flow (\ref{hamiltonian}). Thus a solution eventually entering the negative energy region $$ \{H(u,v)<0\}$$ will converge to $$ (0,\pm\sqrt{m-\omega})$$ spiraling toward that point. A proof of this property follows along the same lines of the analogous one given in \cite{estebannodari}.  This is illustrated by the following picture:

 \begin{figure}[!h]
        \centering
        \includegraphics[scale=.5]{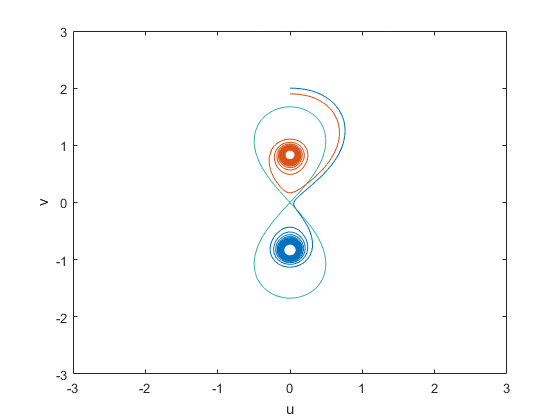}
 \caption{The energy level $\{H=0\}$ and two solutions entering the negative energy set $\{H<0\}$.}
 \label{figura}
   \end{figure}

If, on the contrary, there holds $$H_{\lambda}(r)>0, \qquad\forall r>0$$ then necessarily the solution tends to the origin, thus corresponding to a localized solution of our PDE.


\subsection{The shooting method}
In our proof we use some ideas from \cite{letreustnodari}, \cite{troy}.

 \begin{definition}
 Put $I_{-1}=\emptyset$. For $k\in\mathbb{N}$ we define 
 
 \begin{equation}
 \label{sets}
 \begin{aligned}
   A_{k}&= \left\{\lambda>0 : \lim_{r\rightarrow+\infty}H_{\lambda}(r)<0, v_{\lambda}\:\mbox{changes sign $k$ times on}\: (0,+\infty)\right\}\\ 
   I_{k}&=  \left\{\lambda>0 : \lim_{r\rightarrow+\infty}(u_{\lambda}(r),v_{\lambda}(r))=(0,0), v_{\lambda}\: \mbox{changes sign $k$ times on} \:(0,+\infty)\right\}.
\end{aligned}
\end{equation}
 \end{definition}

It is immediate so see that $$A_{0}\neq\emptyset$$ as it includes the interval $\left(0,\sqrt{2(m-\omega}\right]$, since $$\{0\}\times\left(0,\sqrt{2(m-\omega)}\right]\subseteq\{(u,v)\in\R^{2} : H(u,v)\leq0\}$$ 

Moreover, numerical simulations indicates that the set $A_{0}$ is bounded and that $A_{1}$ is non-empty and unbounded. This implies that $I_{0}$ is non empty (see \ref{core}). Solutions tending to the origin are expected to appear in the shooting procedure when $\lambda$ passes from $A_{k}$ to $A_{k+1}$, as in (Figure \ref{figura}).

\begin{remark}
We found no numerical evidence for the existence of excited states. This may lead to conjecture that there are no nodal solutions, that is $I_{k}=\emptyset$ for $k>1$. 
The absence of excited states is compatible with the bubbling phenomenon (see the Introduction), which might prevent the existence of those solutions. However in 3D Lorentz-invariant models (\cite{els},\cite{exciteddirac}) it is known that they exist. 
\end{remark}
In this section we show that $I_{0}$ is non empty, thus proving (Theorem \ref{main}). 
This will be achieved in several intermediate steps.

We start with some preliminary lemmas, which are an adaptation of analogous results from \cite{letreustnodari}.
\begin{lemma}
\label{decay}
Let $(u_{\lambda},v_{\lambda})$ be a solution of (\ref{radial}) such that $v_{\lambda}$ changes sign a finite number of times and 
$$ \lim_{r\rightarrow +\infty}H_{\lambda}(r)\geq 0$$
then 
\begin{equation}
\vert u_{\lambda}(r) \vert+\vert v_{\lambda}(r)\vert\leq C e^{-\left(\frac{m-\omega}{2}\right)r}\quad,\quad\forall r\geq 0
\end{equation}
and thus
$$ \lim_{r\rightarrow+\infty}(u_{\lambda}(r),v_{\lambda}(r))=(0,0)$$
\end{lemma}
\begin{proof}
We start by showing that under the above assumptions there exists $\overline{R}\in(0,+\infty)$ such that
\begin{equation}
u_{\lambda}(r)v_{\lambda}(r)>0 \quad,\quad\forall r\geq\overline{R}
\end{equation}
Since $v_{\lambda}$ changes sign a finite number of times, we may suppose w.l.o.g. that for some $ R>0$ 
$$  v_{\lambda}(r)>0\quad,\quad\forall r\geq R$$

We have to prove that $\exists R<\overline{R}<+\infty$ such that
$$ u_{\lambda}(r)>0\quad,\quad\forall r\geq \overline{R}$$ 
 Assume, by contradiction, that 
 $$u_{\lambda}(r)<0\quad,\quad\forall r>R $$
 Then the second equation of (\ref{radial}) implies that $ \dot{v}_{\lambda}(r)>0,\forall r>R$, and $v_{\lambda}$ is increasing for $r>R$. Thus 
$$ \lim_{r\rightarrow+\infty}v_{\lambda}(r)=\delta\in(0,+\infty]$$

Indeed, we cannot have $\delta=+\infty$ as in that case $$ \lim_{r\rightarrow+\infty}H_{\lambda}(r)=+\infty$$ contradicting the fact that $H_{\lambda}$ is decreasing along solutions of (\ref{radial}). 

Let $(\rho_{n})_{n}\subseteq\mathbb{R}$ be a sequence such that $$\lim_{n\rightarrow+\infty}\rho_{n}=+\infty\quad,\quad\lim_{n\rightarrow+\infty}u_{x}(\rho_{n})=\lambda$$
for some $\lambda\in\mathbb{R}$, and consider the solution $(U,V)$ of (\ref{hamiltonian}) such that $$(U(0),V(0))=(\lambda,\delta)$$
 
 By (\ref{stability}), it follows that $(u_{\lambda}(\rho_{n}+\ast),v_{\lambda}(\rho_{n}+\ast))$ converges uniformly to $(U,V)$ on bounded intervals. Since $$ \lim_{n\rightarrow+\infty}v_{\lambda}(\rho_{n}+r)=\delta\quad,\quad\forall r>0$$
 we have $V(r)=\delta$, for any $r\geq0$. The second equation of (\ref{hamiltonian}) implies that $U(r)=0$ for all $r>0$. 
 
 We conclude that $(U,V)$ is an equilibrium of the hamiltonian flow (\ref{hamiltonian}). Since $\delta>0$, $$(\lambda,\delta)=(0,\sqrt{m-\omega}) $$
 
 This is absurd, since we would have $$ 0\leq\lim_{r\rightarrow+\infty}H_{\lambda}(r)\leq H\left(0,\sqrt{m-\omega}\right)<0$$
 Thus there exists $\overline{R}\in(R,+\infty)$ such that $u_{\lambda}(\overline{R})=0$. Note that we have $$\dot{u}_{\lambda}(\overline{R})>0$$
 Indeed, $\dot{u}_{\lambda}(\overline{R})=v_{x}(\overline{R})\left[v^{2}_{\lambda}(\overline{R})-(m-\omega)\right]>0$
where the term in the r.h.s. is positive, otherwise the point $(0,v_{\lambda}(\overline{R}))$ would belong to the negative energy region, contradicting our assumptions on $H_{\lambda}(r)$.

Now suppose that there exists $R<\overline{R}<R'$ such that $u_{\lambda}(R')=0$ and ${u}_{\lambda}(r)>0$ on $(\overline{R},R')$. This implies that $\dot{u}_{\lambda}$ is negative in a left neighborhood of $R'$. By the first equation of (\ref{radial}), we get $$ v^{2}_{\lambda}(R')-(m-\omega)\leq0$$
Then $(0,v_{\lambda}(R'))\in\{H(u,v)<0\}$, and this is absurd as already remarked.

We thus conclude that 
\begin{equation}
\label{positive}
u_{\lambda}(r)>0\quad,\quad\forall r\geq\overline{R} 
\end{equation}

The second equation of (\ref{radial}) shows that $v_{\lambda}$ is decreasing on $(\overline{R},+\infty)$ and by (\ref{stability}), arguing as above, it can be proved that
\begin{equation}
\label{zero}
\lim_{r\rightarrow+\infty}(u_{\lambda}(r),v_{\lambda}(r))=(0,0)
\end{equation}
We now prove the exponential decay.

By (\ref{radial}), (\ref{zero}), (\ref{positive}), we have for all $r>\overline{R}$ 
 
\begin{equation}
\left\{\begin{aligned}
    \dot{u}_{\lambda} &\leq \frac{(m-\omega)}{2}v_{\lambda}-(m-\omega)v_{x}  \\ 
   \dot{v}&\leq -(m+\omega)u_{\lambda}
\end{aligned}\right.
\end{equation} 
 Then $$\frac{d}{dr}(u_{\lambda}+v_{\lambda})\leq-\frac{(m-\omega)}{2}(u_{\lambda}+v_{\lambda}) $$ for all $r>\overline{R}$. Then the claim follows, since $$u_{\lambda}(r),v_{\lambda}(r)>0,\quad \forall r\geq\overline{R}.$$
 \end{proof}

 \begin{lemma}
 \label{universal}
 There exists a constant $C_{0}>0$ such that, if for some $R>1$
 \begin{enumerate}
 \item $H_{\lambda}(R)<\frac{C_{0}}{R}$; 
  \item $u_{\lambda}(R)v_{\lambda}(R)>0$ and $v^{2}_{\lambda}(R)<2(m-\omega)$;
  \item $v_{\lambda}\:\mbox{changes sign $k$ times on}\:[0,R]$;
\end{enumerate}
then $\lambda \in A_{k} \cup I_{k} \cup A_{k+1}$.
 \end{lemma}
 \begin{proof}
 Suppose, by contradiction, that $\lambda\notin A_{k} \cup I_{k} \cup A_{k+1}$. 
 
 W.l.o.g. we can assume that $u_{\lambda}(R)>0$ and $v_{\lambda}(R)>0$. Let
 $$\overline{R}:=\inf\{r>R\: :\: u_{\lambda}(r)\leq0\}\in(R,+\infty] $$ 
 Note that $v_{\lambda}$ changes sign exactly once in $(R,\overline{R})$. Indeed, as long as $u_{\lambda}>0$ the second equation of (\ref{radial}) shows that $v_{\lambda}$ is decreasing. Moreover we cannot have $v_{\lambda}(r)>0$ for all $(R,\overline{R})$, as in that case the solution would enter the negative energy zone or tend to the origin. This is impossible, since $\lambda\notin A_{k}\cup I_{k}$.
 
 Now suppose that $\overline{R}=+\infty$. We have seen that there exists $R<\tilde{R}<+\infty$ such that $v_{\lambda}<0$ on $(\tilde{R},\overline{R})$. Arguing as in the proof of (Lemma \ref{decay}), one easily sees that 
 $$\lim_{r\rightarrow+\infty}v_{\lambda}(r)=\delta\in(-\infty,0) $$
 Moreover, the solution tends to an equilibrium $(\lambda,\delta)$ of the hamiltonian system (\ref{hamiltonian}), as $r\longrightarrow+\infty$. 
 
 Thus $(\lambda,\delta)=(0,-\sqrt{m-\omega})$, giving a contradiction,as 
  $$0\leq\lim_{r\rightarrow+\infty}H_{\lambda}(r)=H\left(0,-\sqrt{(m-\omega)}\right)<0 $$ 
 Then $\overline{R}<+\infty$ and we have 
 $$u_{\lambda}(\overline{R})=0\qquad,\qquad v_{\lambda}(\overline{R})\leq-\sqrt{2(m-\omega)} $$
 since we must have $H_{\lambda}(\overline{R})>0$.
 
Let $R<R_{1}<R_{2}<\overline{R}$ be such that 

\begin{equation}
\label{decrease}
 v_{\lambda}(R_{1})=-\frac{\sqrt{m-\omega}}{2}\quad ,\quad v_{\lambda}(R_{2})=-\sqrt{m-\omega}
\end{equation}
 
Since $R>1$, we have $H_{\lambda}(R)<C_{0}$ and if $C_{0}$ is sufficiently small we have that 
\begin{equation}
\label{bounded}
u_{\lambda}(r)\leq\sqrt{m-\omega}\quad,\quad\forall r\in[R_{1},R_{2}] 
\end{equation}
We have, since $v_{\lambda}$ is decreasing and by (\ref{radial},\ref{bounded},\ref{decrease}) 
$$\frac{\sqrt{m-\omega}}{2}=v_{\lambda}(R_{1})-v_{\lambda}(R_{2})=-\int^{R_{2}}_{R_{1}}\dot{v}_{\lambda}(r)dr=\int^{R_{2}}_{R_{1}}\sqrt{m-\omega}\left(3m-\omega\right)dr $$
and then
\begin{equation}
\label{lower}
(R_{2}-R_{1})\geq\frac{1}{2(3m-\omega)}
\end{equation}
Moreover, a simple computation gives 
\begin{equation}
\label{trick}
\frac{1}{r}\frac{d}{dr}\left(r^{2}H_{\lambda}(r)\right)=2H_{\lambda}(r)+r\dot{H}_{\lambda}(r)=-\frac{u^{4}_{\lambda}(r)}{2}+\frac{v^{2}_{\lambda}(r)}{2}\left[v^{2}_{\lambda}(r)-2(m-\omega)\right]
\end{equation}
and then
\begin{equation}
\label{trick2}
\frac{d}{dr}\left(r^{2}H_{\lambda}(r)\right)<0\quad ,\quad \forall r\in[R,R_{2}]
\end{equation}
By (\ref{decrease},\ref{trick}) we have
\begin{equation}
\label{small}
\begin{split}
(R_{2})^{2}H_{\lambda}(R_{2})-(R_{1})^{2}H_{\lambda}(R_{1})&\leq-\int^{R_{2}}_{R_{1}}\frac{(m-\omega)^{2}}{2}rdr\\&=-\frac{(m-\omega)^{2}}{4}(R_{2}+R_{1})(R_{2}-R_{1})\\& \leq-\frac{(m-\omega)^{2}}{4(3m-\omega)}R
\end{split}
\end{equation}
Since the map $r\mapsto r^{2}H_{\lambda}(r)$ is decreasing on $[R,R_{2}]$ by (\ref{trick2}), then (\ref{small}) implies that
\begin{equation}
\begin{split}
(R_{2})^{2}H_{\lambda}(R_{2})&\leq(R_{1})^{2}H_{\lambda}(R_{1})-\frac{(m-\omega)^{2}}{4(3m-\omega)}R \\&\leq R^{2}\left(H_{\lambda}(R)-\frac{(m-\omega)^{2}}{4R(3m-\omega)}\right)\leq0
\end{split}
\end{equation}
if $C_{0}\leq \frac{(m-\omega)^{2}}{4(3m-\omega)}$. Then $$H_{\lambda}(R_{2})\leq 0 $$
reaching a contradiction, and the lemma is proved.
 \end{proof}
 
 The next lemma gives the main properties of the sets $A_{k}$ and $I_{k}$.
 
 \begin{lemma}
 \label{core}
 For all $k\in\mathbb{N}$ we have
 \begin{enumerate}
 \item $A_{k}$ is an open set;
 \item if $\lambda\in I_{k}$ then there exists $\varepsilon>0$ such that $(\lambda-\varepsilon,\lambda+\varepsilon)\subseteq A_{k}\cup I_{k}\cup A_{k+1}$;
 \item if $A_{k}\neq\emptyset$ and it is bounded, we have $\sup A_{k}\in I_{k}$;
 \item if $I_{k}\neq\emptyset$ and it is bounded, then $\sup I_{k}\in I_{k}$.
 \end{enumerate}
 \end{lemma}
 \begin{proof}
 \begin{enumerate}
 
 \item It follows from the continuity of the flow (\ref{radial}) w.r.t. the initial datum (Lemma \ref{existence});
 \item Let $\lambda\in I_{k}$. By Lemma (\ref{decay}) $$\vert u_{\lambda}(r)\vert+\vert v_{\lambda}(r)\vert\leq C\exp\left(-\frac{m-\omega}{2}r\right)\quad,\quad\forall r\geq0 $$ and then, given $C_{0}>0$ as in Lemma (\ref{trick}), $\exists R>1$ such that $H_{\lambda}(R)<\frac{C_{0}}{R}$, $u_{\lambda}(R)v_{\lambda}(R)>0$ and $v_{\lambda}$ changes sign $k$ times on $[0,R]$. 
 
 The continuity of the flow (\ref{radial}) implies that the same holds for an initial datum $y\in(\lambda-\varepsilon,\lambda+\varepsilon)$ for $\varepsilon>0$ small. The claim then follows by Lemma (\ref{trick}).
 \item Let $\lambda=\sup A_{k}$ and $(\lambda_{i})\subseteq A_{k}$ such that $\lim_{i\rightarrow+\infty}\lambda_{i}=\lambda$. 
 
 If we suppose that  $\lambda\in A_{r}$ for some $r\in\mathbb{N}$, then by continuity of the flow we also have $\lambda_{i}\in A_{r}$, for $i$ large. 
 This implies that $r=k$, that is, $ \lambda\in A_{k}$ which is absurd because $A_{k}$ is an open set, by point $(1)$. 
 
 Thus there holds $ \lambda\in I_{s}$, for some $s\in \mathbb{N}$, and by point $(2)$ there exists $\eps>0$ such that $$ \lambda\in A_{s}\cup I_{s}\cup A_{s+1}$$
 which implies that the same holds for $\lambda_{i}$, provided $i$ is large. Then, as before,  we have $s=k $.
 
 Moreover, as already remarked $$\lambda\notin \bigcup_{j\in\mathbb{N}}A_{j} $$
 and then the claim follows.
 \item Arguing as in the proof of point $(3)$ we get that $$ \sup I_{k}\in I_{r}$$
 for some $r\in\mathbb{N}$. Then we conclude as before, using point $(2)$.
 \end{enumerate}
 \end{proof}
 
 We want to prove that the set $A_{0}$ is bounded, showing that if $\lambda>0$ is large enough then there exists $R_{\lambda}>0$ such that $v_{\lambda}(R_{\lambda})=0$, as strongly suggested by numerical simulations (see Figure \ref{figura2}).
  \begin{figure}[!t]
  \label{figura2}
        \centering
        \includegraphics[scale=.4]{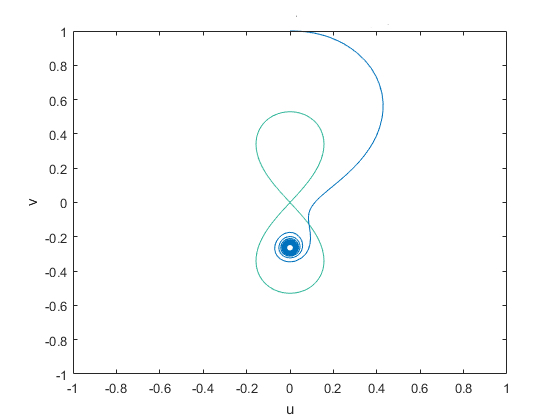}
 \caption{A solution entering the lower half-plane $\{v<0\}$.}
   \end{figure}

 To do so we relate solutions corresponding to such data to those of a limiting problem, inspired by \cite{troy}.
 
 \subsection{Asymptotic expansion}\label{asymptotic}
In this section we provide, after a suitable scaling, a precise asymptotic expansion that will allow us to control the behavior of the solution in term of the initial datum.

 Put $\eps=\lambda^{-1}$ and consider the following rescaling
 
 \begin{equation}\label{rescaling}
\left\{\begin{aligned}
    U_{\eps}(r) &= \eps u_{\lambda}(\eps^{2}r)  \\ 
   V_{\eps}(r)&=\eps v_{\lambda}(\eps^{2}r)
\end{aligned}\right.
\end{equation}

Using (\ref{radial}) we find the system for $(U_{\eps},V_{\eps})$:
 
 \begin{equation}\label{rescaled}
\left\{\begin{aligned}
    \dot{U}_{\eps}+\frac{U_{\eps}}{r} &= (U^{2}_{\eps}+V^{2}_{\eps})V_{\eps}-\eps^{2}(m-\omega)V_{\eps}  \\ 
   \dot{V}_{\eps}&= -(U^{2}_{\eps}+V^{2}_{\eps})U_{\eps}-\eps^{2}(m+\omega)U_{\eps}
\end{aligned}\right.
\end{equation}
together with the initial conditions $U_{\eps}(0)=0$, $V_{\eps}(0)=1$.

The limiting problem as $\eps\rightarrow 0$ (and thus $\lambda\rightarrow +\infty$) is 

\begin{equation}\label{limiting}
\left\{\begin{aligned}
    \dot{U}_{0}+\frac{U_{0}}{r} &= (U^{2}_{0}+V_{0}^{2})V_{0}  \\ 
   \dot{V}_{0}&= -(U^{2}_{0}+V^{2}_{0})U_{0}
\end{aligned}\right.
\end{equation}
with $U_{0}(0)=0$, $V_{0}(0)=1$.

As in \cite{Aubin} we consider the family of spinors given by
\begin{equation}\label{family} 
\varphi(y)=f(y)(1-y)\cdot\varphi_{0} \qquad y\in\mathbb{R}^{2}
\end{equation} 
where $\varphi_{0}\in \mathbb{C}^{2}$, $f(y)=\frac{2}{1+\vert y\vert^{2}}$ and the dot represents the Clifford product.

It can be easily checked that they are $\mathring{H}^{\frac{1}{2}}(\R^{2},\C^{2})$-solutions to the following Dirac equation
\begin{equation}\label{conformal}
\mathcal{D}\varphi=\vert\varphi\vert^{2}\varphi
\end{equation} 

\begin{remark}
The spin structure of euclidean spaces is quite explicit and the spinors given in (\ref{family}) can be rewritten in matrix notation, as 
$$\varphi(y)=f(y)(\mathds{1}_{2}+iy_{1}\sigma_{1}+iy_{2}\sigma_{2})\cdot\varphi_{0} \qquad y\in\mathbb{R}^{2}$$
$\mathds{1}_{2}$ and $\sigma_{i}$ being the identity and the Pauli matrices, respectively. 

See \cite{Jost} for more details.
\end{remark}
A straightforward (but tedious) computation shows that the spinors defined in (\ref{family}) are of the form of the ansatz (\ref{ansatz}), thus being solutions to the system (\ref{limiting}). Exploiting the conformal invariance of (\ref{conformal}) (see \cite{Isobecritical}) one can easily see that the solution matching the above initial conditions is
 
 \begin{equation}\label{massless}
\left(U_{0}(r)=\frac{2r}{4+r^{2}}, V_{0}(r)=\frac{4}{4+r^{2}}\right)
 \end{equation} 
 
 \begin{lemma}\label{uniform}
We have $$(U_{\eps},V_{\eps})\xrightarrow{\eps\rightarrow 0} (U_{0},V_{0})$$ uniformly on $[0,T]$, for all $T>0$, where $(U_{0},V_{0})$ is the solutions to the limiting problem (\ref{limiting}).
 \end{lemma}
 
 \begin{proof}
Fix $T>0$ and let $r\in[0,T]$.

Remark that the system (\ref{rescaled}) is equivalent to

\begin{equation}\label{integrallambda}
\left\{\begin{aligned}
   U_{\eps}(r) &= \frac{1}{r}\int^{r} _{0}sV_{\eps}(s)[U_{\eps}^{2}(s)+V_{\eps}^{2}(s)-\eps^{2}(m-\omega)]ds \\ 
   V_{\eps}(r)&=1 -\int^{r}_{0}U_{\eps}(s)[(U^{2}_{\eps}(s)+V^{2}_{\eps}(s))+\eps^{2}(m+\omega)]ds
\end{aligned}\right.
\end{equation}

Similarly, we can rewrite (\ref{limiting}) as 
\be\label{integrallimiting}
\left\{\begin{aligned}
   U_{0}(r) &= \frac{1}{r}\int^{r} _{0}sV_{0}(s)(U_{0}^{2}(s)+V_{\eps}^{2}(s))ds \\ 
   V_{0}(r)&=1 -\int^{r}_{0}U_{0}(s)(U^{2}_{0}(s)+V^{2}_{0}(s))ds
\end{aligned}\right.\ee

Arguing as for (\ref{radial}), for each fixed $\eps>0$ we associate a hamiltonian to the system (\ref{rescaled})
$$\tilde{H}_{\eps}(U,V):= \frac{(U^{2}+V^{2})^{2}}{4}+\eps^{2}\frac{m}{2}(U^{2}-V^{2})+\eps^{2}\frac{\omega}{2}(U^{2}+V^{2})$$ 
It's easy to see that $\tilde{H}_{\eps}$ is decreasing along the flow, so that $$\tilde{H}_{\eps}(U_{\eps}(r),V_{\eps}(r))\leq \tilde{H}_{\eps}(0,1)\leq 1\qquad\forall r\geq0.$$
The coercivity of $H_{\eps}$ then implies that 
\be\label{bounded}
\vert U_{\eps}(r)\vert +\vert V_{\eps}(r)\vert\leq C \qquad\forall r\geq0
\ee
for some $C>0$ independent of $\eps$.

By (\ref{integrallambda},\ref{integrallimiting}) and since $r\in[0,T]$ we get
\be
\begin{split}
\vert U_{\eps}(r)-U_{0}(r)\vert+&\vert V_{\eps}(r)-V_{0}(r)\vert\leq \int^{r}_{0}\left\vert V_{\eps}(V^{2}_{\eps}+U^{2}_{\eps})- V_{0}(V^{2}_{0}+U^{2}_{0})\right\vert ds\\ &+ \int^{r}_{0}\left\vert U_{\eps}(V^{2}_{\eps}+U^{2}_{\eps})-U_{0}(V^{2}_{0}+U^{2}_{0}) \right\vert ds+2\eps^{2}mT
\end{split}
\ee
It's not hard to see that the first two integrands in the r.h.s of the above inequality are locally Lipschitz. Then by (\ref{bounded}) we have
\be
\vert U_{\eps}(r)-U_{0}(r)\vert+\vert V_{\eps}(r)-V_{0}(r)\vert\lesssim \int^{r}_{0}\left(\vert U_{\eps}-U_{0}\vert+\vert V_{\eps}-V_{0}\vert\right) ds+2\eps^{2}mT
\ee

Since $r\in[0,T]$, the Gronwall lemma gives
\be
\vert U_{\eps}(r)-U_{0}(r)\vert+\vert V_{\eps}(r)-V_{0}(r)\vert\lesssim\eps^{2}
\ee
thus proving the claim.
 \end{proof}
 
 The above results is not enough to conclude that $V_{\eps}$ changes sign, since $V_{0}>0$ for all $r\geq0$. 
 
 We obtain a more refined analysis of the behavior of the solution thanks to a continuity argument.
 
 We consider the solution $(U_{\eps},V_{\eps})$ as a perturbation of $(U_{0},V_{0})$, as follows:
 \be\label{perturbation} 
 \left\{\begin{aligned}
    U_{\eps}(r)&=U_{0}(r)+\eps^{2}h_{1}(r)+\eps^{4}h_{2}(r,\eps)  \\ 
    V_{\eps}(r)&=V_{0}(r)+\eps^{2}k_{1}(r)+\eps^{4}k_{2}(r,\eps)
\end{aligned}\right.
 \ee
   
 and substituting into (\ref{rescaled}) we get the following linear system for $\eps^{2}$-order terms
 \be\label{one}
 \left\{\begin{aligned}
    \dot{h}_{1}+\frac{h_{1}}{r}&=-(m-\omega)V_{0}+2U_{0}V_{0}h_{1}+(U^{2}_{0}+3V^{2}_{0})k_{1}  \\ 
    \dot{k}_{1}&=-(m+\omega)U_{0}-2U_{0}V_{0}k_{1}-(3U^{2}_{0}+V^{2}_{0})h_{1}
\end{aligned}\right.
 \ee
 and we impose the initial conditions 
 \be
 h_{1}(0)=0\qquad,\quad k_{1}(0)=0
 \ee
 Rewriting (\ref{one}) in integral form, as in (\ref{integral}), we have:
 \be\label{estimatesone}
 \left\{\begin{aligned}
    \vert{h}_{1}(r)\vert&\leq\int^{r}_{0}(m-\omega)V_{0}ds+\int^{r}_{0}\left[2U_{0}V_{0}\vert h_{1}\vert+(U^{2}_{0}+3V^{2}_{0})\vert k_{1}\vert\right]ds  \\ 
    \vert{k}_{1}(r)\vert&\leq\int^{r}_{0}(m+\omega)U_{0}ds+\int^{r}_{0}\left[2U_{0}V_{0}\vert k_{1}\vert+(3U^{2}_{0}+V^{2}_{0})\vert h_{1}\vert\right]ds
\end{aligned}\right.
 \ee
 Remark that $$U_{0}V_{0}(r) + V^{2}_{0}(r)\leq U^{2}_{0}(r)\leq V_{0}\quad,\quad \forall r>2$$  and that
$$V_{0}\in L^{1}(\R^{+}) .$$
 
 Moreover, there holds
 \be\label{logaritmo} 
 U_{0}(r)=\frac{2r}{4+r^{2}}\sim \frac{2}{r}\qquad\mbox{as}\quad r\longrightarrow+\infty.
 \ee
 
Then summing up both sides of (\ref{estimatesone}) we get:
 \be
 \vert h_{1}(r)\vert+\vert k_{1}(r)\vert \lesssim \int^{r}_{0}U_{0}ds+\int^{r}_{0}\left(\vert h_{1}\vert+\vert k_{1}\vert\right)V_{0}ds
 \ee
 The Gronwall inequality thus gives:
 \be\label{gronwallone}
 \vert h_{1}(r)\vert+\vert k_{1}(r)\vert\lesssim\left( \int^{r}_{0}U_{0}ds\right)\exp\left(C\int^{r}_{0}V_{0}ds\right)\lesssim  \int^{r}_{0}U_{0}ds
 \ee
 where $C>0$ is a constant.

By (\ref{logaritmo}) we can thus conclude that 
 \be\label{logone}
 \vert h_{1}(r)\vert+\vert k_{1}(r)\vert\lesssim \ln(r)\qquad\mbox{as}\quad r\longrightarrow+\infty
 \ee
The above estimates imply that $$2U_{0}V_{0}k_{1}, (3U^{2}_{0}+V^{2}_{0})h_{1}\in L^{1}(\R^{+})$$ and then integrating the second equation in (\ref{one}) we get
\be\label{loggrowth}
h_{1}(r)\sim -\ln(r) \qquad\mbox{as}\quad r\longrightarrow+\infty
\ee
 We now have to deal with remainder terms in (\ref{perturbation}). 
 
 In particular, we want to analyze the behavior of those terms on the time interval $\left(0,\frac{1}{\eps}\right)$, thanks to a continuity argument based on the Gronwall inequality.
 
 Let 
 \be\label{sup}
 \overline{r}_{\eps}:=\sup\left\{ r\in\left[0,\eps^{-1}\right) : \vert h_{2}(r,\eps)\vert+\vert k_{2}(r,\eps)\vert < \eps^{-\frac{3}{2}} \right\}
 \ee
 Since $h_{2}(0,\eps)=k_{2}(0,\eps)=0$, by continuity it's evident that $$ \overline{r}_{\eps}>0$$
 As shown in the Appendix using the equations for $h_{2}$ and $k_{2}$ one gets the following estimates:
 \be\label{estimatessecond}
\vert h_{2}(r,\eps)\vert +\vert k_{2}(r,\eps)\vert \lesssim\frac{1}{\eps}\ln\left(\frac{1}{\eps}\right)+\int^{r}_{0}\left(V_{0}(s)+\eps^{2}\right)\left(\vert h_{2}(s,\eps)\vert+\vert k_{2}(s,\eps)\vert\right) ds  
 \ee
 for $0<r<\overline{r}_{\eps}\leq\frac{1}{\eps}$.
 
 The Gronwall estimates then gives:
 \be
 \vert h_{2}(r,\eps)\vert +\vert k_{2}(r,\eps)\vert \lesssim \frac{1}{\eps}\ln\left(\frac{1}{\eps}\right)\exp\left(C\int^{r}_{0}\left(\eps^{2}+V_{0}(s)\right)ds\right)
 \ee
 for some $C>0$.
 
Since $r<\frac{1}{\eps}$ and $V_{0}\in L^{1}(\R^{+})$ we eventually have:

\be\label{estimatestwo}
\vert h_{2}(r,\eps)\vert +\vert k_{2}(r,\eps)\vert \lesssim \frac{1}{\eps}\ln\left(\frac{1}{\eps}\right)
\ee

Now, if we suppose that 
$$ \overline{r}_{\eps} < \frac{1}{\eps}$$

by (\ref{estimatestwo}) and by continuity there exists $\delta>0$ such that 

$$\frac{1}{\eps}\ln\left(\frac{1}{\eps}\right) \lesssim \vert h_{2}(r,\eps)\vert +\vert k_{2}(r,\eps)\vert \leq\eps^{-\frac{3}{2}} $$ 
 
for all $r\in [\overline{r}_{\eps},\overline{r}_{\eps}+\delta)$, thus contradicting the definition in (\ref{sup}). 

Then there holds:
\be
\vert h_{2}(r,\eps)\vert+\vert k_{2}(r,\eps)\vert < \eps^{-\frac{3}{2}} ,\qquad \forall r\in\left(0,\frac{1}{\eps}\right)
\ee

Recall that the second equation in (\ref{perturbation}) reads as 
$$ V_{\eps}(r)=V_{0}(r)+\eps^{2}k_{1}(r)+\eps^{4}k_{2}(r,\eps)$$

By (\ref{estimatestwo}) and (\ref{massless}) we see that $$V_{0}=O(\eps^{2})\quad,\quad k_{2}=o(\eps^{2})\qquad \mbox{as}\quad r\rightarrow \left(\frac{1}{\eps}\right)^{-}$$

Then by (\ref{loggrowth}) we get

\be
V_{\eps}(r)\sim -\eps^{2}\ln{\frac{1}{\eps}}\qquad\mbox{as}\quad r\rightarrow \left(\frac{1}{\eps}\right)^{-}
\ee 
Thus we have
\be
V_{\eps}(R_{\eps})=0\qquad \mbox{for some}\quad  R_{\eps}\in \left(0,\frac{1}{\eps}\right)
\ee
 
 In view of the scaling (\ref{rescaling}), we conclude that for large initial data $\lambda>0$, the corresponding solution $(u_{\lambda},v_{\lambda})$ of (\ref{radial}) has at least one node. 
 
 This proves the following (recall the definition (\ref{core}))

\begin{lemma}
 The set $A_{0}$ is bounded.
 \end{lemma}
 
 Then by (Lemma \ref{core}) we have that $I_{0}\neq\emptyset$, that is the system (\ref{radial}) admits a solution without nodes, tending to $(0,0)$ as $r\rightarrow+\infty$, which correspond to a localised solution of equation (\ref{equation}). The exponential decay follows by (Lemma \ref{decay}). This proves (Theorem \ref{main}). 

\section*{Appendix}
In this section we prove the estimates (\ref{estimatessecond}) for remainder terms in (\ref{perturbation}).  

For the sake of brevity we only deal with $k_{2}$. The estimate for $h_{2}$ follows along the same lines with obvious modifications.

Inserting the ansatz (\ref{perturbation}) into the system (\ref{rescaled}), using equations (\ref{massless}) and (\ref{one}) and imposing the initial condition we get the following equation  

 \be\label{two}
 \left\{\begin{aligned}
    \frac{d}{dr}k_{2}(r,\eps)&=K_{0}(r)+\eps^{2}K_{2}(r)+\eps^{4}K_{4}(r)+\eps^{6}K_{6}(r)+\eps^{8}K_{8}(r) \\ 
    k_{2}(0,\eps)&=0
\end{aligned}\right.
 \ee
for all $\eps>0$. Note that the $K^{i}$s do not depend on $\eps$.

The terms in the r.h.s. are given by

\be\label{resto}
 \left\{\begin{aligned}
    K_{0}&=-\left(2U^{2}_{0}+V^{2}_{0}+2U_{0}V_{0}\right)h_{2}-U_{0}\left(3h^{2}_{1}+k^{2}_{1}\right)-2V_{0}h_{1}k_{1}-(m+\omega)h_{1} \\ 
    K_{2}&=-U_{0}\left(4h_{1}h_{2}+2k_{1}k_{2}\right)-(h^{3}_{1}+h_{1}k^{2}_{1})-2V_{0}(h_{1}k_{2}+k_{1}h_{2})-(m+\omega)h_{2}\\
    K_{4}&=-\left(U_{0}(2h^{2}_{2}+k^{2}_{2})+2V_{0}h_{2}k_{2}\right)-(2h_{1}k_{1}k_{2}+2h^{2}_{1}h_{2}+k^{2}_{1}h_{2})\\
    K_{6}&=-h_{1}h^{2}_{2}-k_{1}k^{2}_{2}-h_{1}h^{2}_{2}-k_{1}k_{2}h_{2}\\
    K_{8}&=-h^{3}_{2}-h_{2}k^{2}_{2}
\end{aligned}\right.
 \ee
Our aim is to estimate $\vert k_{2}(r,\eps)\vert$ for $0<r<\overline{r}$ (see (\ref{sup})) and $0<\eps\ll1$. 

This is achieved integrating (\ref{two}) and estimating the integral of the absolute value of each term in (\ref{resto}).

Remark that, by the definition of $(U_{0},V_{0})$, (\ref{massless})
\be\label{elleuno}
2U^{2}_{0}+V^{2}_{0}+2U_{0}V_{0}\leq V_{0}\in L^{1}(\R^{+})
\ee
Moreover, (\ref{logone}) and (\ref{logaritmo}) imply that $U_{0}\left(3h^{2}_{1}+k^{2}_{1}\right)\notin L^{1}(\R^{+})$ and then
\be
\int^{r}_{0}U_{0}\left\vert3h^{2}_{1}+k^{2}_{1}\right\vert ds\lesssim \int^{\frac{1}{\eps}}_{1}\frac{\ln(s)}{s}ds\lesssim \eps^{-\frac{1}{4}}
\ee
By the above remarks and (\ref{logone}), we have 
\be
V_{0}h_{1}k_{1}\in L^{1}(\R^{+})
\ee 
and 
\be
\int^{r}_{0}\vert h_{1}\vert ds = O(r\ln(r)),\qquad\mbox{as}\quad r\rightarrow+\infty
\ee

Collecting the above esimates we get
\be\label{zero}
\int^{r}_{0}\vert K_{0}\vert ds \lesssim \int^{r}_{0}V_{0}\vert h_{2}\vert ds+ \eps^{-1}\vert\ln(\eps)\vert 
\ee

The second term is estimated as follows.

Recall that 
\be\label{bound}
\vert h_{2}(r)\vert +\vert k_{2}(r)\vert \leq \eps^{-\frac{3}{2}}
\ee
for $0<r\leq \overline{r}$. Then by (\ref{elleuno}) we have

\be
\int^{r}_{0}U_{0}\left\vert4h_{1}h_{2}+2k_{1}k_{2}\right\vert ds\lesssim\eps^{-\frac{3}{2}} \int^{\frac{1}{\eps}}_{1}\frac{\ln(s)}{s} ds\lesssim\eps^{-\frac{7}{4}} 
\ee
Using again (\ref{logone}), it's not hard to see that 

\be
\int^{r}_{0}\vert h^{3}_{1}+h_{1}k^{2}_{1} \vert ds\lesssim \eps^{-\frac{5}{4}}
\ee

Since $$ V_{0}h_{1},V_{0}k_{1}\in L^{1}(\R^{+})$$ by (\ref{bound}) we have 
\be
\int^{r}_{0}V_{0}\left(\vert h_{1}k_{2}\vert+\vert k_{1}h_{2}\vert\right) ds\lesssim\eps^{-\frac{3}{2}} 
\ee
We then conclude that 
\be\label{due}
\int^{r}_{0}\vert K_{2}\vert ds \lesssim \eps^{-\frac{7}{4}} +\int^{r}_{0}\vert h_{2}\vert ds
\ee

Let's turn to the third term.

By (\ref{elleuno}) and (\ref{bound}) and since $U_{0}(r)=\frac{2r}{4+r^{2}}$, we get

\be
\int^{r}_{0}\vert U_{0}(2h^{2}_{2}+k^{2}_{2})+2V_{0}h_{2}k_{2}\vert ds\lesssim \eps^{-3}\int^{\frac{1}{\eps}}_{1}U_{0}ds\lesssim \eps^{-3}\vert\ln(\eps)\vert
\ee

Using (\ref{logone}) and (\ref{bound}) we can estimate
\be\label{quattro}
\int^{\overline{r}}_{0}\vert K_{4}\vert ds \lesssim   \eps^{-3}\vert\ln(\eps)\vert+\eps^{-\frac{5}{4}}\lesssim  \eps^{-3}\vert\ln(\eps)\vert
\ee

All the terms appearing in $K_{6}$ have the same behavior, so that by (\ref{logone}),(\ref{bound}) and above estimates it's easy to see that 

\be\label{sei}
\int^{r}_{0}\vert K_{6}\vert ds \lesssim \eps^{-4}\vert\ln(\eps)\vert
\ee

Lastly, by (\ref{bound}) we can estimate
\be\label{otto}
\int^{r}_{0}\vert K_{8}\vert ds \lesssim \eps^{-\frac{11}{2}}
\ee

Combining (\ref{zero},\ref{due},\ref{quattro},\ref{sei},\ref{otto}), integrating (\ref{resto}) gives 
\be
\vert k_{2}(r,\eps)\vert \lesssim \eps^{-1}\vert\ln(\eps)\vert + \int^{r}_{0}(V_{0}(s)+\eps^{2})\vert h_{2}(s,\eps)\vert ds
\ee

Analogous estimates can be worked out for $h_{2}$, obtaining
\be
\vert h_{2}(r,\eps)\vert \lesssim \eps^{-1}\vert\ln(\eps)\vert + \int^{r}_{0}(V_{0}(s)+\eps^{2})\vert k_{2}(s,\eps)\vert ds
\ee

and the claimed inequality (\ref{estimatessecond}) follows by summing up the last two estimates.
\section*{References}

\bibliographystyle{elsarticle-num}

\bibliography{ShootingDirac}
 
\end{document}